\documentclass[a4paper, 11pt]{article}
\usepackage[utf8]{inputenc}
\usepackage[top=1in, bottom=1.2in, left=1in, right=1in]{geometry}
\usepackage{amsmath, amssymb, amsthm}
\usepackage{enumerate}
\usepackage{hyperref}

\newcommand{\T}{{\rm T}}
\newcommand{\cc}{\mathbb{C}}
\newcommand{\h}{\mathcal{H}}
\newcommand{\s}{\mathcal{S}}
\newcommand{\bra}[1]{\left\langle #1 \right|}
\newcommand{\ket}[1]{\left| #1 \right\rangle}
\newcommand{\braket}[1]{\left\langle #1 \right\rangle}

\newtheorem{definition}{Definition}[section]
\newtheorem{proposition}{Proposition}[section]
\newtheorem{theorem}[proposition]{Theorem}
\newtheorem{lemma}[proposition]{Lemma}

\numberwithin{equation}{section}

\title{{\bf A Subspace of Maximal Dimension with\\ Bounded Schmidt Rank}}
\author{Priyabrata Bag \& Santanu Dey}

\date{}

\begin{document}
\maketitle

\begin{abstract}
\noindent We study Schmidt rank for a vector (i.e., a pure state) and Schmidt number
for a mixed state which are entanglement measures. We show that if a subspace of a
certain bipartite system contains no vector of Schmidt rank $\leqslant k$, then any
state supported on that space has Schmidt number at least $k+1$. A construction of
subspace of $\cc^m \otimes \cc^n$ of maximal dimension, which does not contain any
vector of Schmidt rank less than $3$, is given here.\\
\mbox{}
\mbox{}
\end{abstract}

\section{Introduction}
\label{intro}

A \emph{state} (also known as the density matrices) is a positive matrix whose trace is equal to 1.
A pure state $\ket{v} \in \cc^m \otimes \cc^n$ (or $\ket{v} \bra{v}$) is said to be a
\emph{product state} if it can be written as
$\ket{v}=\ket{v_1}\otimes \ket{v_2}=\ket{v_1} \ket{v_2}$ for some $\ket{v_1} \in \cc^m$ and
$\ket{v_2} \in \cc^n$; otherwise it is called \emph{entangled}. The states, which are not pure,
are also referred to as \emph{mixed states}. If $\{\ket{e_0},\ket{e_1}\}$ is the standard basis of
$\cc^2,$ then the states
\[\frac{\ket{e_0} \ket{e_0}\pm\ket{e_1} \ket{e_1}}{\sqrt{2}}
\mbox{~and~}\frac{\ket{e_0} \ket{e_1}\pm\ket{e_1} \ket{e_0}}{\sqrt{2}}\]
are famous entangled states (cf. Sections 1.3.7 and 2.3 of \cite{NC13:QCQI}), which are known as
\emph{Bell states} or \emph{EPR pairs}.

A mixed state $\rho\in M_m\otimes M_n$ is called \emph{separable} if it is convex combination of
pure product states; otherwise it is called \emph{entangled}.
Entanglement is the key property of quantum systems which is responsible for the higher
efficiency of quantum computation and tasks like teleportation, super-dense coding, etc
(cf. \cite{HHHH09:QE}). The Schmidt rank of vectors and Schmidt number
of states in a bipartite finite dimensional Hilbert space are measures of entanglement. In
Section~\ref{SR&SN}, we establish a sufficient condition for states on a bipartite finite
dimensional Hilbert space to be of Schmidt number bounded below by some positive integer. A
construction of subspace of $\cc^m \otimes \cc^n$ of maximal dimension, which does not contain any
vector of Schmidt rank less than $3$, is given in Section~\ref{SMDSR}.

In \cite{CMW08:DSBSR}, it was proved  using algebraic geometric techniques that for a bipartite
system $\cc^m \otimes \cc^n,$ the dimension of any subspace of Schmidt rank greater than or equal
to $k$ is bounded above by $(m-k+1)(n-k+1).$  In \cite[Proposition~10]{CMW08:DSBSR} it was also shown
that this bound is attained. In this article, motivated by the analysis done in \cite{Bha06:CESMD},
we present an alternate approach to prove that this bound is attained for $k=3.$ We construct a
subspace of $\cc^m \otimes \cc^n$ of Schmidt rank greater than or equal to 3 and, unlike
\cite{CMW08:DSBSR}, we also have a basis of this subspace consisting of elements of Schmidt rank 3.
For the case when a subspace of $\cc^m\otimes\cc^n$ is of Schmidt rank  greater than or equal to 2
(i.e., the subspace does not contain any product vector), the maximum dimension of that subspace
is  $(m-1)(n-1),$ and this was first proved  in \cite{Par04:OMDCES} and \cite{Wal02:UGT}
(cf. \cite{Bha06:CESMD}).

\section{Schmidt Rank and Schmidt Number}
\label{SR&SN}

Let $\h$ denote the bipartite Hilbert space $\cc^m\otimes\cc^n$. By Schmidt decomposition theorem
(cf. \cite[Section~2.5]{NC13:QCQI}), any pure state $\ket{\psi}\in\h$ can be written as
\begin{equation}
 \ket{\psi}=\sum_{j=1}^k\alpha_j\ket{u_j}\otimes\ket{v_j} \label{sd}
\end{equation}
for some $k\leqslant\min\{m,n\}$, where $\{\ket{u_j}:1\leqslant j\leqslant k\}$ and
$\{\ket{v_j}:1\leqslant j\leqslant k\}$ are orthonormal sets in $\cc^m$ and $\cc^n$ respectively,
and $\alpha_j$'s are nonnegative real numbers satisfying $\sum_j\alpha_j^2=1$.
\begin{definition}
 In the Schmidt decomposition \eqref{sd} of a pure bipartite state $\ket{\psi}$ the minimum number
 of terms required in the summation is known as the Schmidt rank of $\ket{\psi}$, and it is
 denoted by $SR(\ket{\psi})$.
\end{definition}

In the bipartite Hilbert space $\cc^m\otimes\cc^n,$ for any $1 \leqslant r \leqslant\min\{m,n\},$
there is at least some state $\ket{\psi}$ with $SR(\ket{\psi})=r.$ Any state $\rho$ on a finite
dimensional Hilbert space $\h$ can be written as
\begin{equation}
 \rho=\sum_jp_j\ket{\psi_j}\bra{\psi_j}, \label{spd}
\end{equation}
where  $\ket{\psi_j}$'s are pure states in $\h$ and $\{p_j\}$ forms a probability
distribution. The following notion was introduced in \cite{TH00:SNDM}:
\begin{definition}
 The Schmidt number of a state $\rho$ on a bipartite finite dimensional Hilbert space $\h$ is defined
 to be the least natural number $k$ such that $\rho$ has a decomposition of the form given in 
 \eqref{spd} with $SR(\ket{\psi_j})\leqslant k$ for all $j$. The Schmidt number of $\rho$ is
 denoted by $SN(\rho)$.
\end{definition}

The next theorem gives a sufficient condition for a state to have Schmidt number greater than $k.$
\begin{theorem}
 Let $\s$ be a subspace of $\h=\cc^m\otimes\cc^n$ which does not contain any vector of Schmidt rank
 lesser or equal to $k$. Then any state $\rho$ supported on $\s$ has Schmidt number at least $k+1$.
\end{theorem}
\begin{proof}
Let $\rho$ be a state with Schmidt number, $SN(\rho)=r\leqslant k$. So, $\rho$ can be
written as \[\rho=\sum_jp_j\ket{v_j}\bra{v_j},\] where $SR(\ket{v_j})\leqslant k$ for all $j$ and
$\{p_j\}$ forms a probability distribution. First we show that each of $\ket{v_j}$ is in the range
of $\rho$. With out loss of generality, we show this for $j=1$. Write 
\begin{equation} \label{rqt}
\rho=p_1\ket{v_1}\bra{v_1}+T,
\end{equation}
where $p_1>0$ and $T$ is a nonnegative operator.
 
Let $\ket{\psi}\neq 0$ in $\h$ be such that $T\ket{\psi}=0$ and $\braket{v_1|\psi}\neq 0$. Clearly,
$\rho\ket{\psi}$ is a nonzero multiple of $\ket{v_1}$ and hence $\ket{v_1}$ belongs to the range of
$\rho$.
Now suppose the null space of $T$ is contained in $\{\ket{v_1}\}^\perp$. Then $\ket{v_1}$ is in the
range of $T$. So, there exists $\ket{\psi}\neq 0$ such that 
$T\ket{\psi}=\ket{v_1}.$ Thus, \eqref{rqt} yields
$$\rho\ket{\psi}=(p_1\braket{v_1|\psi}+1)\ket{v_1}.$$ 
If $\rho \ket{\psi}=0,$ then by the positivity of 
$\rho, p_1 \ket{v_1}\bra{v_1}$ and $T,$ and \eqref{rqt}, we obtain $\braket{T \psi|\psi}=0,$ 
and hence $T \ket{\psi}=0.$ This is a contradiction. Thus, we deduce that $\rho \ket{\psi} \neq 0.$ 
Therefore, it follows that $\ket{v_1}$ is in the range of $\rho$.
 
If such a $\rho$ is supported on $\s$, then the above statement gives a contradiction to the fact
that $\s$ does not contain any vector of Schmidt rank lesser or equal to $k$.
This completes the proof.
\end{proof}

\section{A Subspace of Maximal Dimension of\texorpdfstring{\\}{ }
Schmidt Rank \texorpdfstring{$\geqslant 3$}{>=3}}
\label{SMDSR}

Let $\h=\cc^m \otimes \cc^n$ as before. Let us fix an infinite dimensional Hilbert space
${\mathcal K}$ with orthonormal basis $\{\ket{e_0},\ket{e_1},\ldots \},$ and identify $\cc^m$ and
$\cc^n$ with $\mbox{span} \{\ket{e_0},\ket{e_1},\ldots,\ket{e_{m-1}}\}$ and
$\mbox{span} \{\ket{e_0},\ket{e_1},\ldots,\ket{e_{n-1}}\},$
respectively. Define $N=n+m-2,$ and  for $2 \leqslant d \leqslant N-2$ define
\begin{align}
 \s^{(d)}=\mbox{span} \{\ket{e_{i-1}}\otimes\ket{e_{j+1}} -2\ket{e_i}\otimes\ket{e_j}
 +\ket{e_{i+1}}\otimes\ket{e_{j-1}}:\notag\\
 1 \leqslant i \leqslant m-2,
 1 \leqslant j \leqslant n-2, i+j=d  \}, \label{SR2I}
\end{align}
\begin{equation}
 \s^{(0)}= \s^{(1)}= \s^{(N-1)}= \s^{(N)}=\{0\}, \mbox{~and~}
 \s:= \bigoplus\limits_{d=0}^N \s^{(d)}. \label{SR2II}
\end{equation}
We claim that $\s$ does not contain any vector of Schmidt rank less than 3.
\begin{lemma}\label{ptnm}
The columns of the $(t+2)\times t$ matrix
\[A_t=\left[ \begin{array}{ccccc}
1 & 0&\ldots&0&0\\
-2 & 1 &\ldots &0 &0\\
1  & -2 &\ldots &0&0\\
0 & 1 &\ldots & 0& 0\\
\vdots&\vdots&&\vdots&\vdots\\
0&0&\ldots&1 & 0\\
0&0&\ldots&-2 & 1 \\
0&0&\ldots&1  & -2 \\
0&0&\ldots&0 & 1
\end{array}\right]\]
are linearly independent such that any linear combination of these columns has at least 3
nonzero entries.
\end{lemma}
\begin{proof}
At first, we show that  all the order-$t$ minors of $A_t$ are nonzero. We would prove this
statement by induction. The statement clearly holds for $l=1,$ i.e., for the matrix
$A_1=\left[ \begin{array}{c}
1\\-2\\ 1 \end{array}\right].$
Assume that the statement holds for $l=t-1.$ We would show that the statement is true for
$l=t$ by observing the following cases:
\begin{enumerate}[1.]

\item If the two rows deleted from $A_t$ to obtain the order-$t$ minor does not include the first
row, then let us denote the $t \times t$ matrix by $F$ whose determinant gives this order-$t$
minor. By two elementary row operations (involving the first row) on $F$  which does not change
the determinant, we can obtain  $(1, 0,0, \ldots,0)^{\T}$ as the first column of $F,$ where
superscript $\T$ denotes the transpose. The
order-$(t-1)$ minor obtained by deleting the first row and the first column of $F$ is nonzero
by induction hypothesis because it is a order-$(t-1)$ minor of $A_{t-1}.$ It follows that the
$\det F$ is nonzero.

\item If the two rows deleted from $A_t$ to obtain the order-$t$ minor does not include the last
row, then let us denote the $t \times t$ matrix by $L$ whose determinant gives this order-$t$
minor. By two elementary row operations (involving the last row) on $L$  which does not change the
determinant, we can obtain  $(0, \ldots,0,0,1)^{\T}$ as the last column of $L$. The order-$(t-1)$
minor obtained by deleting the last row and the last column of $F$ is nonzero by induction
hypothesis because it is a order-$(t-1)$ minor of $A_{t-1}.$ It follows that the $\det L$
is nonzero.

\item If the first and the last row is deleted from $A_t$ to obtain the  order-$t$ minor,
then we need to show that the determinant of
\[C_s=\left[ \begin{array}{ccccc} 
-2 & 1 &\ldots &0 &0\\
1  & -2 &\ldots &0&0\\
0 & 1 &\ldots & 0& 0\\
\vdots&\vdots&&\vdots&\vdots\\
0&0&\ldots&1 & 0\\
0&0&\ldots&-2 & 1 \\
0&0&\ldots&1  & -2 
\end{array}\right]\]
is nonzero. We claim that $\det C_s= (-1)^s (s+1).$ Once again induct on $s.$ Verifying for
$s=2,3$ directly. Assume that the formula holds for all $k<s$ and $s>3.$ Then expand along
the first row to see that
\begin{align*}
 \det C_s =& -2(\det C_{s-1})-\det C_{s-2}\\
 =& -2(-1)^{s-1} s -(-1)^{s-2}(s-1)\\
 =& (-1)^{s}(s+1).
\end{align*}
\end{enumerate}
Thus, all the order-$t$ minors of $A_t$ are nonzero. Let us assume that a linear combination of
columns of $A_t$ has less than 3 non zero entries, i.e., a linear combination of columns of $A_t$
has $t$ or more zero entries. Let $I$ be the set of indices of any $t$ of those zero entries.
So, the $t \times t$ submatrix formed from the rows of $A_t$ which are indexed by $I,$  is such
that the columns of this submatrix is linearly dependent. We deduce that the corresponding
order-$t$ minor is zero and this is a contradiction. This proves that any linear combination of
the columns of $A_t$ has at least 3 nonzero entries. It also follows that the columns of the matrix
$A_t$ are linearly independent.
\end{proof}

\begin{theorem}\label{mdsr}
Let $m$ and $n$ be  natural numbers such that $3 \leqslant \min \{m,n\}.$ The space $\s$ defined
by equations \eqref{SR2I} and \eqref{SR2II} does not contain any vector
of Schmidt rank $\leqslant 2$  and  $\dim\s=(m-2)(n-2)$.
\end{theorem}
\begin{proof}
Let us fix an infinite dimensional Hilbert space ${\mathcal K}$ with orthonormal basis
$\{\ket{e_0},\ket{e_1},\ldots \},$ and identify $\cc^m$ and $\cc^n$ with
$\mbox{span} \{\ket{e_0},\ket{e_1},\ldots,\ket{e_{m-1}}\}$ and 
$\mbox{span} \{\ket{e_0},\ket{e_1},\ldots, \ket{e_{n-1}}\},$
respectively, as done before. If for any element
$|v \rangle \in \cc^m \otimes \cc^n,$ we have  $|v \rangle=\sum_{ij} c_{ij} |e_i \rangle \otimes
|e_j \rangle,$ then we identify the
bipartite system $\cc^m\otimes\cc^n$ with the space $M_{m \times n}$ of $m\times n$ matrices by the 
isometric isomorphism $\phi:\cc^m\otimes\cc^n \longrightarrow M_{m \times n}$ defined by 
$\phi(|v \rangle)= [c_{ij}]_{m \times n}.$
It follows from the standard proof of Schmidt decomposition based on the singular value decomposition
that, an element of $\cc^m\otimes\cc^n$ has Schmidt rank at least $r$, if and only if  the
corresponding $m\times n$ matrix is of rank at least $r$. Also, it is known that a matrix has rank
at least $r$ if and only if it has a nonzero minor of order $r$ (cf. \cite[page~18]{HJ85:MA}). Thus,
it is enough to construct a set
of $(m-2)(n-2)$ linearly independent matrices, which are image under $\phi$ of a basis of $\s,$ such
that any linear combination of these matrices has a nonzero minor of order 3.

Label the anti-diagonals of any $m \times n$ matrix by non-negative integers $k,$ such that the first
anti-diagonal from upper left (of length one) is labelled $k=0$  and value of $k$ increases from upper
left to lower right. Let the length of the $k^{th}$ anti-diagonal be denoted by $|k|$. 
If we assume, without loss of generality, that  $m\leqslant n,$ then the formula for $|k|$ is
\[ 
|k|=\left\{ \begin{array}{ll} k+1 & \mbox{for~} 0 \leqslant k \leqslant m-1\\
m & \mbox{for~}  m-1 \leqslant k \leqslant n-1\\
m+n-(k+1) & \mbox{for~} n-1 \leqslant k \leqslant N. \end{array} \right.
\]
Recall that for $2 \leqslant k \leqslant N-2$, $\s^{(k)}$ is generated by the set
\begin{align*}
 B_k=\{\ket{e_{i-1}} \otimes \ket{e_{j+1}} -2 \ket{e_i} \otimes \ket{e_j} + \ket{e_{i+1}} \otimes
 \ket{e_{j-1}}: 1 \leqslant i \leqslant m-2,\\
 1 \leqslant j \leqslant n-2, i+j=k\}.
\end{align*}

For $2 \leqslant k \leqslant N-2$, let $\widetilde{B}_k$ denote the image of $B_k$ under the map
$\phi$. Note that any
element of $\widetilde{B}_k$ is the matrix obtained from the $m \times n$ zero matrix by
replacing the $k^{th}$ anti-diagonal
by $(0, \ldots,0, 1,-2,1,0, \ldots,0).$ 
For  $2 \leqslant k \leqslant N-2$ and  $t=|k|-2,$ from Lemma~\ref{ptnm}, $\widetilde{B}_k$
is a set of $t$ linearly independent matrices.  Also, by Lemma~\ref{ptnm} it follows that if $M$
is a matrix obtained by taking any linear combination of
matrices from $\widetilde{B}_k,$ then $M$ has at least $3$ nonzero entries in the $k^{th}$
anti-diagonal and the entries of $M$,  other than  those on $k^{th}$ anti-diagonal, are zeros. Since the
determinant of the $3 \times 3$ submatrix with these $3$ nonzero elements in its principal
anti-diagonal is clearly nonzero, any linear combination of the matrices in $\widetilde{B}_k$ has at least
one nonzero order-$3$ minor, thus has rank at least $3.$

Let $\widetilde{B}= \displaystyle \bigcup_{k=2}^{N-2} \widetilde{B}_k$. Since elements from different 
$\widetilde{B}_k$ have different
nonzero anti-diagonal, $\widetilde{B}$ is linearly independent. Thus $B= \displaystyle \bigcup_{k=2}^{N-2} B_k$
is a basis for $\s.$ Let $C$ be a matrix obtained by an arbitrary
linear combination from the elements of $\widetilde{B}$. Let $\kappa$ be the largest $k$ for which the linear
combination involves an element from $\widetilde{B}_k$. The $\kappa^{th}$ anti-diagonal of $C$ has at least
$3$ nonzero elements. Because $\kappa$ labels the bottom-rightmost anti-diagonal of $C$ that contains nonzero
elements, the $3 \times 3$ submatrix of $C,$ with these 3 nonzero elements in
the principal anti-diagonal, has only zero entries in all its anti-diagonals which are below the principal
anti-diagonal. Hence  the $3 \times 3$ submatrix has nonzero determinant. Thus, the rank of $C$ is at least 3.
We conclude that  $\s$  does not contain any vector of Schmidt rank $\leqslant 2.$

The dimension of $\s$ is equal to the cardinality of $B$. Without loss of generality, assume $m\leqslant n$.
Then, the cardinality of $B$ is given by
\begin{align*}
 |B| =& \sum_{k=2}^{N-2} |B_k|=\sum_{k=2}^{N-2} (|k|-2)\\
 =& \sum_{k=2}^{m-2} (k-1)+\sum_{k=m-1}^{n-1} (m-2)+\sum_{k=n}^{m+n-4} (m+n-2-(k+1))\\
 =& (n-m+1)(m-2)+2\sum_{k=2}^{m-2}(k-1)\\
 =& (m-2)(n-2).
\end{align*}
\end{proof}

From the above theorem it follows that the basis 
\begin{align*}
 B=\bigcup_{k=2}^{N-2}\{\ket{e_{i-1}} \otimes \ket{e_{j+1}} -2 \ket{e_i} \otimes \ket{e_j} + \ket{e_{i+1}}
 \otimes \ket{e_{j-1}}: 1 \leqslant i \leqslant m-2,\notag\\
 1 \leqslant j \leqslant n-2, i+j=k\},
\end{align*}
of $\s$ is of Schmidt rank $3$ (i.e., all the elements have Schmidt rank $3$).


\bibliography{SMDBSRRef}
\bibliographystyle{plain}

\noindent Priyabrata Bag\\
Department of Mathematics\\
Indian Institute of Technology Bombay\\
Mumbai, Maharashtra 400076, India\\
E-mail: {\it priyabrata@iitb.ac.in}
\vspace*{0.5cm}\\
Santanu Dey\\
Department of Mathematics\\
Indian Institute of Technology Bombay\\
Mumbai, Maharashtra 400076, India\\
E-mail: {\it santanudey@iitb.ac.in}

\end{document}